\documentclass[a4paper,twocolumn,11pt,amssymb,accepted=2022-06-24]{quantumarticle}
\pdfoutput=1
\usepackage[utf8]{inputenc}
\usepackage[english]{babel}
\usepackage[T1]{fontenc}
\usepackage{amsmath}
\usepackage{hyperref}
\usepackage{amsthm}
\usepackage{tikz}
\usepackage{lipsum}

\newtheorem{theorem}{Theorem}

\begin{document}

\title{Finding out all locally indistinguishable sets of generalized Bell states}

\author{Jiang-Tao Yuan}
\email{jtyuan@hpu.edu.cn}
\affiliation{College of Science, Wuxi University, Wuxi, 214105, China}
\affiliation{School of Mathematics and Information Science, Henan Polytechnic University, Jiaozuo, 454000, China}
\orcid{0000-0002-2960-5199}
\author{Ying-Hui Yang}
\email{yangyinghui4149@163.com}
\affiliation{School of Mathematics and Information Science, Henan Polytechnic University, Jiaozuo, 454000, China}
\orcid{0000-0003-1555-4746}
\author{Cai-Hong Wang}
\email[]{chwang@hpu.edu.cn}
\affiliation{College of Science, Wuxi University, Wuxi, 214105, China}
\affiliation{School of Mathematics and Information Science, Henan Polytechnic University, Jiaozuo, 454000, China}
\orcid{0000-0002-3696-2899}
\maketitle

\begin{abstract}
In general, for a  bipartite quantum system $\mathbb{C}^{d}\otimes\mathbb{C}^{d}$ and an integer $k$ such that $4\leq k\le d$,
there are few necessary and sufficient conditions for local discrimination of sets of $k$ generalized Bell states (GBSs) and
it is difficult to locally distinguish $k$-GBS sets.
The purpose of this paper is  to completely solve the problem of local discrimination of GBS sets in some bipartite quantum systems.
Firstly three practical and effective sufficient conditions are given,
Fan$^{,}$s and Wang et al.$^{,}$s results [Phys Rev Lett 92, 177905 (2004); Phys Rev A 99, 022307 (2019)] can be deduced as special cases of these conditions.
Secondly in $\mathbb{C}^{4}\otimes\mathbb{C}^{4}$, a necessary and sufficient condition for local discrimination of GBS sets is provided,
and a list of all locally indistinguishable 4-GBS sets is provided,
and then the problem of local discrimination of GBS sets is completely solved.
In $\mathbb{C}^{5}\otimes\mathbb{C}^{5}$, a concise necessary and sufficient condition for one-way local discrimination of GBS sets is obtained,
which gives an affirmative answer to the case $d=5$ of the problem proposed by Wang et al.
\end{abstract}

\section{Introduction}
In quantum mechanics, any set of orthogonal quantum states can be discriminated globally.
In general, for a bipartite composite system $\mathbb{C}^{d}\otimes\mathbb{C}^{d}$,
local operations and classical communication (LOCC)
is not sufficient to distinguish sets of orthogonal quantum states \cite{benn1999pra,walg2000prl,walg2002prl,gho2001prl},
that is, it is a difficult task to locally distinguish sets of orthogonal states.
Any two orthogonal states can be perfectly (or deterministically) distinguished by LOCC \cite{walg2000prl},
a complete set of orthogonal maximally  entangled states (MESs) is locally indistinguishable (deterministically or probabilistically),
and $d+1$ or more MESs in $\mathbb{C}^{d}\otimes\mathbb{C}^{d}$ are locally indistinguishable \cite{horo2003prl,fan2004prl,fan2007pra,gho2004pra,nath2005jmp},
and any three generalized Bell States (GBSs) in $\mathbb{C}^{d}\otimes\mathbb{C}^{d}\ (d\geq 3)$ are always locally distinguishable \cite{wang2017qip}.
The nonlocal nature of quantum information is revealed when a set of orthogonal states
of a composite quantum system cannot be perfectly distinguished by LOCC.
The study of local quantum state discrimination is widely
applied to many fields such as data hiding, quantum secret
sharing, and quantum private query \cite{divi2002trans-inform}-\cite{wei2018trans-comp}.
It is also helpful to explore the power and limitations of LOCC which is of inherent interest as a tool to understand entanglement and  nonlocality.

According to the nature of classical communication,
there is a special case of LOCC: one-way LOCC.
In one-way LOCC, the classical communication is originated by only one fixed part and no information is allowed to move in the other direction.
Similarly a general (full) LOCC means local operations and classical communication, that is,
Alice and Bob can communicate classically as much as they like and iteratively adapt their measurements as they go.
It is natural to regard (full) LOCC protocol as superior to one-way LOCC protocol in terms of state discrimination.
Since we can give an explicit mathematical description of one-way LOCC \cite{nath2013pra},
but can't give a clear mathematical description of (full) LOCC,
most of known LOCC distinguishable sets of orthogonal states are one-way LOCC distinguishable sets.

It is natural to ask whether or not a set of $k(4\leq k\leq d)$ orthogonal MESs in $\mathbb{C}^{d}\otimes\mathbb{C}^{d}$
can be (perfectly) distinguished by LOCC \cite{band2011njp}.
Many advances have been made in this field.
The researchers have constructed many examples of locally indistinguishable GBS sets
\cite{band2011njp,yu2012prl,zhang2015pra,wang2016qip,yang2018pra,yuan2019qip},
for example, Bandyopadhyay et al. \cite{band2011njp} showed one-way LOCC indistinguishable 4-GBS sets
in $\mathbb{C}^{4}\otimes\mathbb{C}^{4}$ and $\mathbb{C}^{5}\otimes\mathbb{C}^{5}$.
On the other hand, some sufficient conditions for local discrimination of GBSs are obtained
by using admissible solution set and nonadmissible solution set \cite{wang2019pra,yang2021qip}.

Although many advances have been made,
they are mainly concrete examples of local indistinguishability and sufficient conditions for local distinguishability.
There are few necessary and sufficient conditions for  local discrimination of orthogonal state sets,
which makes it difficult to locally distinguish  a general GBS set.
Even for a general GBS set in a low dimensional bipartite system,
it is difficult to judge its local distinguishability.

Based on these advances, in this paper, the local discrimination of GBS sets is considered.
Firstly, some practical and effective sufficient conditions for GBS sets in an arbitrary dimensional system  are given.
These sufficient conditions are easy to use and widely applicable, they cover the previous conclusions
including Fan$^{,}$s and Wang et al.$^{,}$s results \cite{fan2004prl,wang2019pra}.
And then we show that these sufficient conditions are also necessary for
local discrimination of all 4-GBS sets in $\mathbb{C}^{4}\otimes\mathbb{C}^{4}$,
and we give a table containing all the locally indistinguishable (standard) 4-GBS sets,
which can be used to judge the local distinguishability of a 4-GBS set quickly.
In $\mathbb{C}^{5}\otimes\mathbb{C}^{5}$, one of the sufficient conditions is proved to be necessary,
thus a necessary and sufficient condition for one-way local discrimination of GBS sets is provided.
The necessary and sufficient condition is concise and computable,
and it gives an affirmative answer to the case $d=5$ of the problem proposed by Wang et al. \cite{wang2019pra,yang2018qip}.

The rest of this paper is organized as follows.
In Section II, we recall some relevant notions and results.
In Section III, we present three suffient conditions for local  discrimination of GBS sets in $\mathbb{C}^{d}\otimes\mathbb{C}^{d}$.
In Section IV, for 4-GBS sets in $\mathbb{C}^{4}\otimes\mathbb{C}^{4}$,
a necessary and sufficient condition for local state discrimination is provided,
and all locally indistinguishable (standard) 4-GBS sets are found.
Further, in Section V, for 4-GBS and 5-GBS sets in $\mathbb{C}^{5}\otimes\mathbb{C}^{5}$,
we show a simple and direct necessary and sufficient condition for one-way local discrimination,
which gives an affirmative answer to  case $d=5$ of  the problem proposed by Wang et al.
Lastly, in Section VI, we draw the conclusions.

\section{Preliminaries}

\newtheorem{definition}{Definition}
\newtheorem{lemma}{Lemma}
\newtheorem{corollary}{Corollary}
\newtheorem{example}{Example}
\newtheorem{proposition}{Proposition}
\newtheorem{problem}{Problem}
\newtheorem{conjecture}{Conjecture}

Consider a $d$-dimensional Hilbert space, $\{|j\rangle\}_{j=0}^{d-1}$ is the computational basis
and $\mathbb{Z}_{d}=\{0,1,\ldots,d-1\}$.
Let $U_{m,n}=X^{m}Z^{n}, m, n\in\mathbb{Z}_{d}$ be generalized Pauli matrices (GPMs)
constituting a basis of unitary operators where $X|j\rangle=|j+1$ mod $d\rangle$, $Z|j\rangle=\omega^{j}|j\rangle$ and $\omega=e^{2\pi i/d}$,
which are generalizations of Pauli matrices.
Denote the set of all GPMs on a $d$-dimensional Hilbert space by $P(d)$.

In a quantum system $\mathbb{C}^{d}\otimes\mathbb{C}^{d}$,
the canonical maximally entangled state $|\Phi\rangle$ is $|\Phi_{00}\rangle=(1/\sqrt{d})\sum_{j=0}^{d-1}|jj\rangle$.
It is known that $(I\otimes U)|\Phi\rangle=(U^{T}\otimes I)|\Phi\rangle$, where $T$ means matrix transposition.
Any MES can be written as $|\Psi\rangle=(I\otimes U)|\Phi\rangle$ where $U$ is unitary.
If $U=X^{m}Z^{n}$, the states
\begin{eqnarray*}
|\Phi_{m,n}\rangle=(I\otimes U_{m,n})|\Phi\rangle
\end{eqnarray*}
are called generalized Bell states (GBSs).
Note that there is a one-to-one correspondence between the MESs and unitaries,
and the corresponding unitaries of GBSs are GPMs.
It is convenient to denote a GBS set $\{ (I\otimes X^{m_i}Z^{n_i})|\Phi\rangle \}$ by $\{  X^{m_i}Z^{n_i} \}$ or $\{  (m_i, n_i) \}$.
For a given GBS set $S=\{X^{m_{j}}Z^{n_{j}}|m_{j}, n_{j}\in\mathbb{Z}_{d}\}$,
the difference set $\Delta S$ of the GBS set $S$ means
\begin{widetext}
  \begin{equation*}
\Delta S=\{(m_{j}-m_{k},n_{j}-n_{k})|(m_j, n_j), (m_k, n_k)\in S, j\neq k\}.
  \end{equation*}
\end{widetext}
It plays an important role in local state discrimination \cite{tian2015pra,yang2018pra}.
This section consists of three parts: Wely commutation relation  and congruence equation,
Clifford operator, and LU-equivalence.

\subsection{Weyl commutation relation and congruence equation}

Two unitaries $A$ and $B$ are called Weyl commutative if
$AB=zBA$ where $z$ is a complex number \cite{petz2008book}.

\begin{lemma}\label{lem2.1}\cite{yuan2020jpa,wang-yuan2021jmp}
For two unitary matrices $A$ and $B$, if they are not commutative and satisfy Weyl commutation relation,
then each eigenvector $u$ of $A$ satisfies $\langle u|B|u\rangle=0$.
Especially, for an arbitrary GBS set $S$,
if there is a GPM $T$ which is not commutative to every GPM $U$ in $\Delta S$,
then each eigenvector $|v\rangle$ of $T$ satisfies $\langle v|U|v\rangle=0$
and the set $S$  is one-way LOCC distinguishable.
\end{lemma}

In order to locally distinguish a GBS set, according to Lemma \ref{lem2.1},
we need to find the GPM $T$ in Lemma \ref{lem2.1}.
For two GPMs $(m,n)$ and $(s,t)$, they are commutative if and only if $ns-mt=0 \mod d$.
So we prepare the following result about congruence equation to accomplish this task.

\begin{lemma}\label{lem2.2}
Let $m, n\in\mathbb{Z}_{d}$, $\gcd(m,n)$ be the greatest common divisor of $m$ and $n$,
and
\begin{eqnarray}\label{cogru2.1}
nx-my=0 \mod d
\end{eqnarray}
be a congruence equation.
Then there is always a solution $(\frac{m}{\gcd(m,n)},\frac{n}{\gcd(m,n)})$ to the congruence equation and the following assertions holds.
\begin{enumerate}
\item[{\rm(1)}] If $m=0,n\neq 0$, then solution set of the equation is
$\{(\frac{d}{\gcd(n,d)}t,k)|k\in\mathbb{Z}_{d},t\in\mathbb{Z}_{\gcd(n,d)}\}$.
\item[{\rm(2)}] If $m \neq 0,n= 0$, then solution set of the equation is
$\{(k,\frac{d}{\gcd(m,d)}t)|k\in\mathbb{Z}_{d},t\in\mathbb{Z}_{\gcd(m,d)}\}$.
\item[{\rm(3)}] If $m\neq0,n\neq 0$, then solution set of the equation is
$\{(\frac{m}{\gcd(m,n)}k+\frac{d}{\gcd(m,n,d)}t,\frac{n}{\gcd(m,n)}k+\frac{d}{\gcd(m,n,d)}t^{\prime})\}$
where $k=0,\cdots,\frac{d}{\gcd(m,n,d)}-1;t,t^{\prime}\in\mathbb{Z}_{\gcd(m,n,d)}$.
\end{enumerate}
\end{lemma}
For a given GPM $(m,n)$ in $\mathbb{C}^{d}\otimes\mathbb{C}^{d}$, there are $d\cdot\gcd(m,n,d)$ solutions to the congruence equation \eqref{cogru2.1}.

\subsection{Clifford operator}

A Clifford operator is defined as a unitary operator which maps the Pauli group to itself under conjugation \cite{hos2005pra,far2014jpa}.
The classical (or symplectic) representation $W$ of a single-qudit Clifford operator
is a unique two by two  symplectic matrix (up to a global phase) over $\mathbb{Z}_d$,
that is,
\begin{eqnarray*}
W = \left[
\begin{array}{llll}
a_1 &b_1\\
a_2 &b_2
\end{array}
\right]
\end{eqnarray*}
where the entries are over $\mathbb{Z}_d$.
It is known that a $2\times 2$ matrix $W$ is symplectic if and only if $\det (W)=a_1b_2-a_2b_1\equiv 1$ (mod $d$).
For a GBS $(m,n)$, the Clifford operator $W$ maps $(m,n)$ to $(a_1m+b_1n, a_2m+b_2n)$,
or $X^mZ^n\mathop{\sim}\limits^W X^{a_1m+b_1n}Z^{a_2m+b_2n}$.
Two basic Clifford operators are  quantum Fourier transform (QFT) gate $R$ and phase-shift gate $P$ are as follows:
\begin{eqnarray*}
R= \left[
\begin{array}{cccc}
0 &-1\\
1 &0
\end{array}
\right],
P = \left[
\begin{array}{llll}
1 &0\\
1 &1
\end{array}
\right].
\end{eqnarray*}
The phase-shift and QFT gates are a necessary and sufficient set of gates to
generate (up to global phase) the entire single-qudit Clifford group in any finite dimension \cite{far2014jpa}.

\subsection{LU-equivalence}
When discussing a large number of GBS sets, we need the concept of local unitary (LU) equivalence.
Let's recall the LU-equivalence of two MES sets.
\begin{definition}[\cite{wu-tian2018pra}]
Let $\{ |\phi_1\rangle , |\phi_2\rangle , \cdots , |\phi_n\rangle \}$ and $\{ |\psi_1\rangle , |\psi_2\rangle , \cdots , |\psi_n\rangle \}$ be two sets of MESs. Their corresponding unitary matrix sets are $M=\{M_1, M_2, \cdots , M_n \}$ and $N=\{N_1, N_2, \cdots , N_n \}$. If there exist two unitary operators $U_A$, $U_B$ and a permutation
$\sigma$ over $\{ 1, 2, \cdots , n \}$ such that $|\psi_i\rangle \approx (U_A \otimes U_B) |\phi_{\sigma (i)}\rangle$,
where $\approx$ denotes ``equal up to some global phase",
then these two MES sets are called LU-equivalent.
Meanwhile the corresponding unitary matrix sets are called $U$-equivalent, that is
\begin{eqnarray*}
LMR\approx N,
\end{eqnarray*}
where $L=U_B$ and $R=U_A^T$, denoted by $M\sim N$.
Especially, when $R=L^\dagger$, the two sets $M$ and $N$ are called unitary conjugate equivalent (UC equivalent),
denoted by $M\mathop{\sim}\limits^L N$.
\end{definition}
In general we can study LU-equivalence of two MES sets by studying the $U$-equivalence of their unitary matrix sets.
For convenience, we call a GBS  set containing the standard MES a standard GBS set,
and a GPM  set containing the identity matrix a standard GPM set.
Let $M=\{M_1, M_2, \cdots , M_n \}$ and $N=\{N_1, N_2, \cdots , N_n \}$ be two GPM sets.
If $M$ is $U$-equivalent to $N$, then there exist two unitaries $L$, $R$  and a permutation $\sigma$ over $\{ 1, 2, \cdots , n \}$
such that $LM_iR \approx N_{\sigma (i)}$.
Thus for a fixed $i$,
\begin{align*}
N_{\sigma (j)}\thickapprox LM_jR=(LM_i)(M_i^\dagger M_j)(LM_i)^\dagger (LM_iR)\\
=(LM_i)(M_i^\dagger M_j)(LM_i)^\dagger N_{\sigma (i)}(j=1,2,\cdots ,n).
\end{align*}
Furthermore, if some $M_i=I\in M$, then we have
\begin{align}\label{lu1}
N_{\sigma (j)} \approx (LM_jL^\dagger )N_{\sigma (i)},
L^\dagger(N_{\sigma (j)}N_{\sigma (i)}^{\dag})L \approx M_j.
\end{align}
The formula \eqref{lu1} tells us that, for any standard GPM set $M$,
if $M$ is U-equivalent to $N$ then $M$ consists of elements of the shape of $L^\dagger(N_{j}N_{i}^{\dag})L$.
In particular, if $L$ is a Clifford transformation,
then the GPM set $M=\{L^\dagger(N_{j}N_{i}^{\dag})L\}$ is U-equivalent to $N$.
For a UC transformation that maps one GPM set $A$ to another GPM set $B$,
if the set $A$ or $B$ containing nondegenerate $X^s$ and $Z^t$,
then the transformation must be a Clifford operator.
These facts lead to the following assertion.

\begin{lemma}\label{lem2.3}
Let $N=\{N_j\}$ be a GPM set, if $L$ is a Clifford transformation and $\sigma$ is a permutation,
then the standard GPM set $M=\{L^\dagger(N_{\sigma (j)}N_{\sigma (i)}^{\dag})L\}$ is U-equivalent to $N$.
Especially, if $N$ contains two nondegenerate GPMs $X^s$ and $Z^t$,
then any standard GPM set $M$ that is U-equivalent to $N$
has the form $M=\{L^\dagger(N_{\sigma (j)}N_{\sigma (i)}^{\dag})L\}$,
where $L$ is a Clifford operator and $\sigma$ is a permutation.
\end{lemma}

For a given GPM set $N$, Lemma \ref{lem2.3} can be used to construct lots of GPM sets which are $U$-equivalent to $N$.
Each $U$-equivalent set of $N$ in Lemma \ref{lem2.3} only depends on Clifford operators and $N$ itself,
and then it gives an effective way to find lots of standard GPM sets that is U-equivalent to $N$.

\section{Sufficient conditions for arbitrary dimensional systems}

In this section, we will show three sufficient conditions for local discrimination of GBS sets
in arbitrary dimensional systems and some examples which can not be determined by previous known sufficient conditions.

In order to find the GPMs in Lemma \ref{lem2.1},
for a GBS set $\mathcal{S}$ and a GPM $(m,n)$ in the difference set $\Delta\mathcal{S}$,
denote the solution set of the corresponding congruence equation \eqref{cogru2.1} by $S(m,n)$,
then every element in the set
\begin{eqnarray*}
\mathcal{D}(\mathcal{S})\triangleq P(d)\setminus \bigcup_{(m,n)\in\Delta\mathcal{S}} S(m,n)
\end{eqnarray*}
satisfies the condition in Lemma \ref{lem2.1} and can be used to  locally distinguish the set $\mathcal{S}$.
Therefore, the set $\mathcal{D}(\mathcal{S})$ can be called a $\it{discriminant\ set}$ of $\mathcal{S}$.

\begin{theorem}\label{th3.1}
Let $\mathcal{S}=\{(m_{i},n_{i})|1\le i\le l\}$ be a $l$-GBS set in  $\mathbb{C}^{d}\otimes\mathbb{C}^{d}$ where $4\le l\le d$,
then the set $\mathcal{S}$ is local distinguishable when any of the following conditions is true.
\begin{enumerate}
\item[{\rm(1)}] The discriminant set $\mathcal{D}(\mathcal{S})$ is not empty.
\item[{\rm(2)}] The set $\Delta\mathcal{S}$ is commutative.
\item[{\rm(3)}] The dimension $d$ is a composite number, and for each $(m,n)\in\Delta\mathcal{S}$, $m$ or $n$ is invertible in $\mathbb{Z}_d$.
\end{enumerate}
\end{theorem}

Since the commutativity of the set $\mathcal{S}$ ensures the commutativity of its difference set $\Delta\mathcal{S}$,
according to Theorem \ref{th3.1} (2), $\mathcal{S}$ is local disdinguishable when $\mathcal{S}$ is commutative.
The condition (3) in Theorem \ref{th3.1} only applies when $d$ is composite.
It is easy to see that, when $d$ is a prime number, the condition (2) is a special case of  (1) in Theorem \ref{th3.1}.

\begin{proof}
From the previous discussion, the conclusion (1) is obvious.
Let's consider (2), since the set $\Delta\mathcal{S}$ is commutative,
there exists a unitary matrix $U$ which maps each GPM $A$ in $\Delta\mathcal{S}$ to a diagonal matrix $U^{\dagger}AU$.
Let $|\beta\rangle\triangleq \frac{1}{\sqrt{d}}|(1,\cdots,1)\rangle$ be a $d$-dimensional unit vector,
then the vector $|\alpha\rangle=U|\beta\rangle$ satisfies
$\langle \alpha|A|\alpha\rangle=\langle \beta|U^{\dagger}AU|\beta\rangle=\hbox{tr}{A}=0$,
thus the set $S$ is one-way LOCC distinguishable.
Now consider the last condition, let $d$ be a composite number and $d=st$ be a decomposition.
For $(m,n)\in\Delta\mathcal{S}$ and $m$ is invertible in $\mathbb{Z}_d$,
the GPM $Z^{t}$ is not commutative to $(m,n)$,
then each eigenstate $|\alpha\rangle$ of $Z^{t}$ satisfies $\langle \alpha|X^{m}Z^{n}|\alpha\rangle=0$.
Similarly, if $n$ is invertible in $\mathbb{Z}_d$,
then each eigenstate $|\beta\rangle$ of $X^{s}$ satisfies $\langle \beta|X^{m}Z^{n}|\beta\rangle=0$.
Since $X^{s}$ and $Z^{t}$ are commutative, they have common eigenstates.
Obviously, each common eigenstate $|\gamma\rangle$ satisfies $\langle \gamma|X^{m}Z^{n}|\gamma\rangle=0$.
Therefore the set $S$ is one-way LOCC distinguishable.
\end{proof}

We give some examples to show that Theorem \ref{th3.1} has a wide range of applications.

For a given GBS set $\mathcal{L}=\{(m_{i},n_{i})|1\le i\le l\le d\}$ in  $\mathbb{C}^{d}\otimes\mathbb{C}^{d}$
and the congruence equation
\begin{eqnarray}\label{congru3.1}
m_{i} y+ n_{i}=m_{j} y+ n_{j} \mod d
\end{eqnarray}
where $i<j$ and $(m_{i},n_{i}),(m_{j},n_{j})\in\mathcal{L}$,
let $S_{ij}$ be the solution set of the equation \eqref{congru3.1},
$S_{\mathcal{L}_{y}}$ be the the union of all  $S_{ij}$,
and $S_{A}$ be $\{0,\cdots,d-1,\infty\}$.
The two sets $S_{A}$ and $S_{\mathcal{L}_{y}}$ are called admissible solution set and nonadmissible solution set
for $\mathcal{L}$ respectively \cite{wang2019pra,yang2021qip}.
If the set $S_{A}\backslash S_{\mathcal{L}_{y}}$ is not empty,
then it is easy to check that, for each element $y$ in $S_{A}\backslash S_{\mathcal{L}_{y}}$,
the GPM $(d-1,y)$ is also an element of the set $\mathcal{D}(\mathcal{L})$,
and then the set $\mathcal{L}$ is local distinguishable.
This shows that Theorem \ref{th3.1} (1) leads to Wang et al.$^{,}$s result
which is an extension of Fan$^{,}$s result \cite[Theorem]{fan2004prl}.
But  Wang et al.$^{,}$s result \cite[Theorem 1]{wang2019pra} is not applicable to the following three examples.

\begin{example}\label{ex3.1}
Consider the 5-GBS set $\mathcal{L}_{1}=\{(0,0),(0,1),(1,0),(1,4),(5,5)\}$
in  $\mathbb{C}^{6}\otimes\mathbb{C}^{6}$.
By using the congruence equation \eqref{congru3.1},
it is easy to check that the set $S_{A}\backslash S_{\mathcal{L}_{1_{y}}}$ is empty and
Wang et al.’s result is not applicable.
Now we apply Theorem \ref{th3.1} (1) to the set $\mathcal{L}_{1}$.
By using congruence equation \eqref{cogru2.1} and simple calculation, the difference set
$\Delta\mathcal{L}_{1}=\{(0,1),(1,0),(1,4),(5,5),(1,5),(1,3),(5,4),(0,4),\\
(4,5), (4,1),(0,5),(5,0),(5,2),(1,1),(5,1),(5,3),\\
(1,2),(0,2),(2,1),(2,5)\}$
and the GPM $(2,3)$ belongs to the discriminant set $\mathcal{D}(\mathcal{L}_{1})$,
and then the set $\mathcal{L}_{1}$ is local distinguishable.
\end{example}

\begin{example}\label{ex3.2}
Consider the GBS set $\mathcal{L}_{2}=\{(1,2),(1,0),(3,2),(3,0)\}$ in  $\mathbb{C}^{4}\otimes\mathbb{C}^{4}$
and the set $\mathcal{L}_{3}=\{(2,3),(2,0),(5,3),(5,0)\}$ in  $\mathbb{C}^{6}\otimes\mathbb{C}^{6}$.
The difference sets are $\Delta\mathcal{L}_{2}=\{(0,2),(2,0),(2,2)\}$ and
$\Delta\mathcal{L}_{3}=\{(0,3),(3,0),(3,3)\}$ respectively.
By Lemma \ref{lem2.2} and the congruence equation \eqref{congru3.1},
it is easy to check that all four sets $S_{A}\backslash S_{\mathcal{L}_{2_{y}}}$, $S_{A}\backslash S_{\mathcal{L}_{3_{y}}}$,
$\mathcal{D}(\mathcal{L}_{2})$ and $\mathcal{D}(\mathcal{L}_{3})$
are empty, and then Wang et al.’s result and Theorem \ref{th3.1} (1) are not applicable.
Fortunately, both sets $\Delta\mathcal{L}_{2}$ and $\Delta\mathcal{L}_{3}$ are commutative.
By Theorem \ref{th3.1} (2), $\mathcal{L}_{2}$ and $\mathcal{L}_{3}$ are local distinguishable.
\end{example}

\begin{example}\label{ex3.3}
Consider the GBS set $\mathcal{L}_{4}=\{(1,2),(1,3),(2,2),(0,1)\}$ in  $\mathbb{C}^{4}\otimes\mathbb{C}^{4}$.
The difference set are $\Delta\mathcal{L}_{4}=\{(0,1),(1,0),(3,3),(1,3),(3,2),(2,3),(0,3),(3,0),\\
(1,1),(3,1),(1,2),(2,1)\}$.
It is easy to check that both sets $S_{A}\backslash S_{\mathcal{L}_{4_{y}}}$ and $\mathcal{D}(\mathcal{L}_{1})$
are empty, and then Wang et al.’s result and Theorem \ref{th3.1} (1)-(2) are not applicable.
Obviously, Theorem \ref{th3.1} (3) is applicable and $\mathcal{L}_{4}$ is local distinguishable.
\end{example}

It is interesting that
all the three examples can not be discriminated by Wang et al.$^{,}$s result.
Example \ref{ex3.1} can only be determined by Theorem \ref{th3.1} (1),
Example \ref{ex3.2} can only be determined by Theorem \ref{th3.1} (2),
and only Theorem \ref{th3.1} (3) applies to Example \ref{ex3.3}.

In this section, on the basis of Weyl commutation relation and congruence equation,
we present three sufficient conditions for local discrimination of GBS sets in arbitrary dimensional systems.
The three conditions show that the problem of local discrimination of the GBS set $\mathcal{S}$
can be solved by the properties of the difference set $\Delta\mathcal{S}$.
For example, the condition (2) says that the commutativity of the elements $(m,n)$ in the difference set of $\mathcal{S}$
leads to that the original set $\mathcal{S}$ is locally distinguishable.

\section{Necessary and sufficient conditions for 4-GBS sets in $\mathbb{C}^{4}\otimes\mathbb{C}^{4}$}

It is obvious that every GBS set is LU-equivalent to a standard GBS  set.
In this section, we will show that if a 4-GBS set in $\mathbb{C}^{4}\otimes\mathbb{C}^{4}$
is local distinguishable then one of the three conditions in Theorem \ref{th3.1} is true,
that is, the three sufficient conditions in Theorem \ref{th3.1} are also necessary.
Moreover, we will provide a list of all local indistinguishable  standard 4-GBS sets.
The list and the necessary and sufficient conditions allow us
to quickly judge the local distinguishability  of any 4-GBS set in $\mathbb{C}^{4}\otimes\mathbb{C}^{4}$.

\begin{theorem}\label{th4.1}
Let $\mathcal{S}=\{(m_{i},n_{i})|1\le i\le 4\}$ be a GBS set in  $\mathbb{C}^{4}\otimes\mathbb{C}^{4}$,
then $\mathcal{S}$ is local distinguishable if and only if one of the following conditions is true.
\begin{enumerate}
\item[{\rm(1)}] The discriminant set $\mathcal{D}(\mathcal{S})$ is not empty.
\item[{\rm(2)}] The set $\Delta\mathcal{S}$ is commutative.
\item[{\rm(3)}] For each $(m,n)\in\Delta\mathcal{S}$, $m$ or $n$ is invertible in $\mathbb{Z}_4$.
\end{enumerate}
\end{theorem}

\begin{proof}
Let $\mathcal{S}=\{(m_{i},n_{i})|1\le i\le 4\}$ be a local distinguishable 4-GBS set,
we need to show that the set  $\mathcal{S}$ meets one of the three conditions.
It is known that all 4-GBS sets can be classified into ten locally inequivalent
classes and the representative elements of these equivalence classes are
\begin{align*}
\mathcal{K}=\{I,X^{2},Z^{2},X^{2}Z^{2}\},&\mathcal{L}=\{I,X,X^{2},X^{3}\},\\
\Gamma^{1}_{20}=\{I,X,Z,X^{2}\},&\Gamma^{1}_{31}=\{I,X,Z,X^{3}Z\},\\
\Gamma^{1}_{33}=\{I,X,Z,X^{3}Z^{3}\},&\Gamma^{2}_{12}=\{I,X,Z^{2},XZ^{2}\},\\
\Gamma^{2}_{30}=\{I,X,Z^{2},X^{3}\},&\Gamma^{1}_{12}=\{I,X,Z,XZ^{2}\},\\
\Gamma^{2}_{20}=\{I,X,Z^{2},X^{2}\},&\Gamma^{2}_{32}=\{I,X,Z^{2},X^{3}Z^{2}\}.
\end{align*}
The first 7 sets are one-way LOCC distinguishable, while
the last 3 sets are one-way LOCC indistinguishable (also two-way LOCC indistinguishable).
Hence the local distinguishable set $\mathcal{S}$ is  LU-equivalent to
one of the first 7 sets \cite{tian2016pra}:
$\mathcal{K},\mathcal{L},\Gamma^{1}_{20},\Gamma^{1}_{31},\Gamma^{1}_{33},\Gamma^{2}_{12},\Gamma^{2}_{30}$.
By the formula \eqref{lu1} and Lemma \ref{lem2.3},
the three conditions in Theorem \ref{th4.1} are invariant under LU-equivalence.
Now we only need to consider  the above seven sets.
The two sets $\mathcal{K}$ and $\mathcal{L}$ are commutative,
then their difference sets are also commutative and they meet the condition (2).
Since the difference set $\Delta\Gamma^{1}_{20}=\{(1,0),(0,1),(2,0),(3,1),(2,3),(3,0),(0,3),(1,3),\\(2,1)\}$,
it is easy to check that $(1,1)\in\mathcal{D}(\Gamma^{1}_{20})$,
and then  $\Gamma^{1}_{20}$ satisfies the condition (1).
Similarly, it is easy to check that $\{(1,1),(1,2)\}\subseteq\mathcal{D}(\Gamma^{1}_{31})$,
$\{(1,1),(1,3)\}\subseteq\mathcal{D}(\Gamma^{2}_{12})$,
and $\{(1,1),(1,3)\}\subseteq\mathcal{D}(\Gamma^{2}_{30})$,
then the three sets also meet the condition (1).
About the set $\Gamma^{1}_{33}$,
since the difference set
$\Delta\Gamma^{1}_{33}=\{(1,0),(0,1),(3,3),(1,3),(2,3),(3,2),(3,0),(0,3),\\(1,1),
(3,1),(2,1),(1,2)\}$,
it is clear that it meets the condition (3).
\end{proof}

Next we will use the LU-equivalence to give a list of all local indistinguishable standard 4-GBS sets.
We first give a protocol according to Lemma \ref{lem2.3} to determine the U-equivalence of two GPM sets $M$ and $N$.

Protocol. Let $M=\{ M_1, \cdots , M_l \}=\{ (s_1, t_1), \cdots , (s_l, t_l) \}$ and
$N=\{ (m_1, n_1), \cdots , (m_l, n_l) \}$ be two sets of GPMs
and $I, X^s, Z^t$ ($s$, $t$ are invertible in $\mathbb{Z} _d$) are contained in $N$.

1. Do $i=1$.

2. Let $L_{i}=\{ (s_1-s_i, t_1-t_i), \cdots , (s_l-s_i, t_l-t_i) \}$.

3. Use a program to apply a Clifford operator
$F= \left[
\begin{array}{cccc}
a_1 &b_1\\
a_2 &b_2
\end{array}
\right]$
 to the set $L_{i}$, then we obtain the set
\begin{widetext}
\begin{eqnarray*}
\{(a_1(s_1-s_i)+b_1(t_1-t_i), a_2(s_1-s_i)+b_2(t_1-t_i)), \cdots ,\\
 (a_1(s_l-s_i)+b_1(t_l-t_i), a_2(s_l-s_i)+b_2(t_l-t_i)) \} (\textrm{mod} \  d)
\end{eqnarray*}
\end{widetext}
for all $0\leq a_1,a_2,b_1,b_2\leq d-1$ and $a_1b_2-a_2b_1 \equiv 1$ (mod $d$).

4. For any $1\le i\le l$ and Clifford operator $F$, print all the sets obtained in step 3,
then we obtain all standard GPM sets that are U-equivalent to $N$.

According to this protocol and the ten 4-GBS sets in \cite[Theorem 2]{tian2016pra},
we get  ten  equivalence classes
$\widetilde{\mathcal{K}},\widetilde{\mathcal{L}},\widetilde{\Gamma^{1}_{20}},\widetilde{\Gamma^{1}_{31}},
\widetilde{\Gamma^{1}_{33}},\widetilde{\Gamma^{2}_{12}},\widetilde{\Gamma^{2}_{30}},
\widetilde{\Gamma^{1}_{12}},\widetilde{\Gamma^{2}_{20}},\widetilde{\Gamma^{2}_{32}}$.
See Table I for the number of  standard GBS sets in each equivalence class.

\begin{table*}[htbp]
\centering
\caption{Numbers of standard GBS sets in equivalence classes}
\label{tab4.1}
\begin{tabular}{|c|c|c|c|c|c|c|c|c|c|c|c|}
\hline
$\widetilde{\mathcal{K}}$&$\widetilde{\mathcal{L}}$&$\widetilde{\Gamma^{1}_{20}}$&$\widetilde{\Gamma^{1}_{31}}$&$\widetilde{\Gamma^{1}_{33}}$
&$\widetilde{\Gamma^{2}_{12}}$&$\widetilde{\Gamma^{2}_{30}}$&$\widetilde{\Gamma^{1}_{12}}$&$\widetilde{\Gamma^{2}_{20}}$&$\widetilde{\Gamma^{2}_{32}}$& Total\\
\hline
1&6&192&48 &16&12 &24&96&48&12&455   \\
\hline
\end{tabular}
\end{table*}

The ten 4-GBS sets are not LU equivalent to each other,
and the ten equivalence classes contain a total of 455 standard GBS sets.
Since there are $C_{15}^{3}(=455$) standard 4-GBS sets in $\mathbb{C}^{4}\otimes\mathbb{C}^{4}$,
the ten equivalence classes make up a complete classification,
and all 4-GBS sets in the  three equivalence classes
$\widetilde{\Gamma^{1}_{12}},\widetilde{\Gamma^{2}_{20}},\widetilde{\Gamma^{2}_{32}}$
constitute all local indistinguishable standard 4-GBS sets in $\mathbb{C}^{4}\otimes\mathbb{C}^{4}$.
That is, there are 156 local indistinguishable standard 4-GBS sets, see Table II for details.
Using Table II, we can immediately determine the distinguishability of an arbitrary 4-GBS set.

\begin{table*}[!ht]
\caption{All local indistinguishable standard 4-GBS sets in $\mathbb{C}^{4}\otimes\mathbb{C}^{4}$ (156 items)}
\label{tab4.1}
\centering
\scalebox{0.85}{
\begin{tabular}{|c|l|}
\hline
$d$ &\\
\hline
&\{(0,0),(0,1),(1,0),(1,2)\},\{(0,0),(0,1),(1,0),(2,1)\},\{(0,0),(0,1),(1,0),(2,2)\},\{(0,0),(0,1),(1,1),(1,3)\},\\
$\widetilde{\Gamma^{1}_{12}}$
&\{(0,0),(0,1),(1,1),(2,0)\},\{(0,0),(0,1),(1,1),(2,3)\},\{(0,0),(0,1),(1,2),(2,1)\},\{(0,0),(0,1),(1,2),(2,2)\},\\
(96 items)
&\{(0,0),(0,1),(1,3),(2,0)\},\{(0,0),(0,1),(1,3),(2,3)\},\{(0,0),(0,1),(2,0),(3,1)\},\{(0,0),(0,1),(2,0),(3,3)\},\\
&\{(0,0),(0,1),(2,1),(3,0)\},\{(0,0),(0,1),(2,1),(3,2)\},\{(0,0),(0,1),(2,2),(3,0)\},\{(0,0),(0,1),(2,2),(3,2)\},\\
&\{(0,0),(0,1),(2,3),(3,1)\},\{(0,0),(0,1),(2,3),(3,3)\},\{(0,0),(0,1),(3,0),(3,2)\},\{(0,0),(0,1),(3,1),(3,3)\},\\
&\{(0,0),(0,2),(1,0),(1,1)\},\{(0,0),(0,2),(1,0),(1,3)\},\{(0,0),(0,2),(1,0),(3,1)\},\{(0,0),(0,2),(1,0),(3,3)\},\\
&\{(0,0),(0,2),(1,1),(1,2)\},\{(0,0),(0,2),(1,1),(3,0)\},\{(0,0),(0,2),(1,1),(3,2)\},\{(0,0),(0,2),(1,2),(1,3)\},\\
&\{(0,0),(0,2),(1,2),(3,1)\},\{(0,0),(0,2),(1,2),(3,3)\},\{(0,0),(0,2),(1,3),(3,0)\},\{(0,0),(0,2),(1,3),(3,2)\},\\
&\{(0,0),(0,2),(3,0),(3,1)\},\{(0,0),(0,2),(3,0),(3,3)\},\{(0,0),(0,2),(3,1),(3,2)\},\{(0,0),(0,2),(3,2),(3,3)\},\\
&\{(0,0),(0,3),(1,0),(1,2)\},\{(0,0),(0,3),(1,0),(2,2)\},\{(0,0),(0,3),(1,0),(2,3)\},\{(0,0),(0,3),(1,1),(1,3)\},\\
&\{(0,0),(0,3),(1,1),(2	0)\},\{(0,0),(0,3),(1,1),(2,1)\},\{(0,0),(0,3),(1,2),(2,2)\},\{(0,0),(0,3),(1,2),(2,3)\},\\
&\{(0,0),(0,3),(1,3),(2,0)\},\{(0,0),(0,3),(1,3),(2,1)\},\{(0,0),(0,3),(2,0),(3,1)\},\{(0,0),(0,3),(2,0),(3,3)\},\\
&\{(0,0),(0,3),(2,1),(3,1)\},\{(0,0),(0,3),(2,1),(3,3)\},\{(0,0),(0,3),(2,2),(3,0)\},\{(0,0),(0,3),(2,2),(3,2)\},\\
&\{(0,0),(0,3),(2,3),(3,0)\},\{(0,0),(0,3),(2,3),(3,2)\},\{(0,0),(0,3),(3,0),(3,2)\},\{(0,0),(0,3),(3,1),(3,3)\},\\
&\{(0,0),(1,0),(1,1),(3,1)\},\{(0,0),(1,0),(1,1),(3,2)\},\{(0,0),(1,0),(1,2),(2,1)\},\{(0,0),(1,0),(1,2),(2,3)\},\\
&\{(0,0),(1,0),(1,3),(3,2)\},\{(0,0),(1,0),(1,3),(3,3)\},\{(0,0),(1,0),(2,1),(2,2)\},\{(0,0),(1,0),(2,2),(2,3)\},\\
&\{(0,0),(1,0),(3,1),(3,2)\},\{(0,0),(1,0),(3,2),(3,3)\},\{(0,0),(1,1),(1,2),(3,0)\},\{(0,0),(1,1),(1,2),(3,1)\},\\
&\{(0,0),(1,1),(1,3),(2,1)\},\{(0,0),(1,1),(1,3),(2,3)\},\{(0,0),(1,1),(2,0),(2,1)\},\{(0,0),(1,1),(2,0),(2,3)\},\\
&\{(0,0),(1,1),(3,0),(3,1)\},\{(0,0),(1,1),(3,1),(3,2)\},\{(0,0),(1,2),(1,3),(3,0)\},\{(0,0),(1,2),(1,3),(3,3)\},\\
&\{(0,0),(1,2),(2,1),(2,2)\},\{(0,0),(1,2),(2,2),(2,3)\},\{(0,0),(1,2),(3,0),(3,1)\},\{(0,0),(1,2),(3,0),(3,3)\},\\
&\{(0,0),(1,3),(2,0),(2,1)\},\{(0,0),(1,3),(2,0),(2,3)\},\{(0,0),(1,3),(3,0),(3,3)\},\{(0,0),(1,3),(3,2),(3,3)\},\\
&\{(0,0),(2,0),(2,1),(3,1)\},\{(0,0),(2,0),(2,1),(3,3)\},\{(0,0),(2,0),(2,3),(3,1)\},\{(0,0),(2,0),(2,3),(3,3)\},\\
&\{(0,0),(2,1),(2,2),(3,0)\},\{(0,0),(2,1),(2,2),(3,2)\},\{(0,0),(2,1),(3,0),(3,2)\},\{(0,0),(2,1),(3,1),(3,3)\},\\
&\{(0,0),(2,2),(2,3),(3,0)\},\{(0,0),(2,2),(2,3),(3,2)\},\{(0,0),(2,3),(3,0),(3,2)\},\{(0,0),(2,3),(3,1),(3,3)\}.\\
\hline
&\{(0,0),(0,1),(0,2),(2,0)\},\{(0,0),(0,1),(0,2),(2,2)\},\{(0,0),(0,1),(0,3),(2,1)\},\{(0,0),(0,1),(0,3),(2,3)\},\\
$\widetilde{\Gamma^{2}_{20}}$
&\{(0,0),(0,1),(2,0),(2,2)\},\{(0,0),(0,1),(2,1),(2,3)\},\{(0,0),(0,2),(0,3),(2,0)\},\{(0,0),(0,2),(0,3),(2,2)\},\\
(48 items)
&\{(0,0),(0,2),(1,0),(2,0)\},\{(0,0),(0,2),(1,0),(2,2)\},\{(0,0),(0,2),(1,1),(2,0)\},\{(0,0),(0,2),(1,1),(2,2)\},\\
&\{(0,0),(0,2),(1,2),(2,0)\},\{(0,0),(0,2),(1,2),(2,2)\},\{(0,0),(0,2),(1,3),(2,0)\},\{(0,0),(0,2),(1,3),(2,2)\},\\
&\{(0,0),(0,2),(2,0),(2,1)\},\{(0,0),(0,2),(2,0),(2,3)\},\{(0,0),(0,2),(2,0),(3,0)\},\{(0,0),(0,2),(2,0),(3,1)\},\\
&\{(0,0),(0,2),(2,0),(3	2)\},\{(0,0),(0,2),(2,0),(3,3)\},\{(0,0),(0,2),(2,1),(2,2)\},\{(0,0),(0,2),(2,2),(2,3)\},\\
&\{(0,0),(0,2),(2,2),(3,0)\},\{(0,0),(0,2),(2,2),(3,1)\},\{(0,0),(0,2),(2,2),(3,2)\},\{(0,0),(0,2),(2,2),(3,3)\},\\
&\{(0,0),(0,3),(2,0),(2,2)\},\{(0,0),(0,3),(2,1),(2,3)\},\{(0,0),(1,0),(1,2),(3,0)\},\{(0,0),(1,0),(1,2),(3,2)\},\\
&\{(0,0),(1,0),(2,0),(2,2)\},\{(0,0),(1,0),(3,0),(3,2)\},\{(0,0),(1,1),(1,3),(3,1)\},\{(0,0),(1,1),(1,3),(3,3)\},\\
&\{(0,0),(1,1),(2,0),(2,2)\},\{(0,0),(1,1),(3,1),(3,3)\},\{(0,0),(1,2),(2,0),(2,2)\},\{(0,0),(1,2),(3,0),(3,2)\},\\
&\{(0,0),(1,3),(2,0),(2,2)\},\{(0,0),(1,3),(3,1),(3,3)\},\{(0,0),(2,0),(2,1),(2,2)\},\{(0,0),(2,0),(2,2),(2,3)\},\\
&\{(0,0),(2,0),(2,2),(3,0)\},\{(0,0),(2,0),(2,2),(3,1)\},\{(0,0),(2,0),(2,2),(3,2)\},\{(0,0),(2,0),(2,2),(3,3)\}.\\
\hline
&\{(0,0),(0,1),(2,0),(2,3)\},\{(0,0),(0,1),(2,1),(2,2)\},\{(0,0),(0,2),(1,0),(3,2)\},\{(0,0),(0,2),(1,1),(3,1)\},\\
$\widetilde{\Gamma^{2}_{32}}$
&\{(0,0),(0,2),(1,2),(3,0)\},\{(0,0),(0,2),(1,3),(3,3)\},\{(0,0),(0,3),(2,0),(2,1)\},\{(0,0),(0,3),(2,2),(2,3)\},\\
(12 items)&\{(0,0),(1,0),(1,2),(2,2)\},\{(0,0),(1,1),(1,3),(2,0)\},\{(0,0),(2,0),(3,1),(3,3)\},\{(0,0),(2,2),(3,0),(3,2)\}.\\
\hline
\end{tabular}}
\end{table*}

In this section, we showed that, in $\mathbb{C}^{4}\otimes\mathbb{C}^{4}$,
the three sufficient conditions are also necessary conditions according to the following three facts:
(a) The three conditions are invariant under LU-equivalence (see formular \eqref{lu1}).
(b) All $C_{16}^{4}$(=1820) 4-GBS sets are classified into ten equivalence classes \cite{tian2016pra}
and every 4-GBS set is LU-equivalent to one of the ten representative elements of these equivalence classes.
Then we only need to consider  the ten 4-GBS sets.
(c) The local discrimination of the ten 4-GBS sets can be determined
by \cite[Corollary 1]{band2011njp}, Theorem \ref{th3.1} and Yu et al.$^{,}$s result \cite[Lemma]{yu2015arxiv}.
Moreover we presented a protocol based on LU-equivalence (Lemma \ref{lem2.3}),
according to which all 156 standard 4-GBS sets are found (see Table \ref{tab4.1}).

\section{Necessary and sufficient conditions for GBS sets in $\mathbb{C}^{5}\otimes\mathbb{C}^{5}$}

In this section, we will show that if a GBS set $\mathcal{L}$ in $\mathbb{C}^{5}\otimes\mathbb{C}^{5}$
is one-way LOCC distinguishable then the condition (1) in Theorem \ref{th3.1} is true.
If the dimension $d$ is a prime,
it is easy to check that the condition (1) in Theorem \ref{th3.1} is  equivalent to a simple conclusion:
the cardinality of the set $\{m^{-1}n|(m,n)\in\Delta\mathcal{L}, m^{-1}n\triangleq \infty \hbox{ if }m=0\}$ is less than $d+1$.
The cardinality can be calculated directly from the difference set without solving the congruence equation \eqref{cogru2.1},
and then the necessary and sufficient conditions allow us
to quickly judge the one-way LOCC distinguishability  of any GBS set in $\mathbb{C}^{5}\otimes\mathbb{C}^{5}$.
For convenience, the set $\{m^{-1}n|(m,n)\in\Delta\mathcal{L}, m^{-1}n\triangleq \infty \hbox{ if }m=0\}$
can be called a $\it{discriminant\ index\ set}$ of $\mathcal{L}$.

\begin{theorem}\label{tm5.1}
Let $\mathcal{L}=\{(m_{i},n_{i})\}$ be a 4-GBS or 5-GBS set in  $\mathbb{C}^{5}\otimes\mathbb{C}^{5}$,
then $\mathcal{L}$ is one-way LOCC indistinguishable if and only if
the discriminant set $\mathcal{D}(\mathcal{L})$ is empty,
or equivalently the cardinality of the discriminant index set
$\{m^{-1}n|(m,n)\in\Delta\mathcal{L}, m^{-1}n\triangleq \infty \hbox{ if }m=0\}$ is 6.
\end{theorem}

Since the dimension $5$ is a prime,
it is easy to check that each element  $m^{-1}n$ in Theorem \ref{tm5.1} corresponds to a number $y\in S_{\mathcal{L}_{y}}$,
and each number $y\in S_{\mathcal{L}_{y}}$ also corresponds to
an equivalent class of commutative pairs defined in \cite{yang2018qip}.
So, we get the conclusion as follows.

\begin{corollary}\label{cor5.1}
Let $\mathcal{L}=\{(m_{i},n_{i})|1\le i\le 5\}$ be a 4-GBS or 5-GBS set in  $\mathbb{C}^{5}\otimes\mathbb{C}^{5}$,
then the following assertions are equivalent to each other.
\begin{enumerate}
\item[{\rm(1)}] $\mathcal{L}$ is one-way LOCC indistinguishable.
\item[{\rm(2)}] The cardinality $|\{m^{-1}n|(m,n)\in\Delta\mathcal{L}, m^{-1}n\triangleq \infty \hbox{ if }m=0\}|$ is 6.
\item[{\rm(3)}] The cardinality $|S_{\mathcal{L}_{y}}|$ of the nonadmissible solution set $S_{\mathcal{L}_{y}}$ is 6.
\item[{\rm(4)}] Let  $S$ be a maximal set containing the pairwise noncommuting pairs in deference set,
then the cardinality $|S|$ is equal to 6.
\end{enumerate}
\end{corollary}

Corollary \ref{cor5.1} implies that the following case $d=5$ of Wang et al' problem \cite[Problem]{wang2019pra} is true.

\begin{problem}\label{prob5.1}
Let $\mathcal{L}=\{X^{m_{i}}Z^{n_{i}}|1\le i\le l\}$ be a $l$-GBS set of dimension $d > 2$.
If $S_{A}\backslash S_{\mathcal{L}_{y}}=\emptyset$, is the set $\mathcal{L}$ indistinguishable by one-way LOCC?
\end{problem}

If $d$ is a prime, Problem \ref{prob5.1} is equivalent to the conjecture posed in \cite{yang2018qip}.

\begin{conjecture}\label{conj5.1}
Let $d$ be a prime. If $|S|=d+1$,
then the set $\mathcal{L}$ cannot be distinguished by one-way LOCC.
\end{conjecture}

When $d$ is a composite number,  it obvious that Problem \ref{prob5.1} is not true (see Examples \ref{ex3.1}-\ref{ex3.3}).
When $d$ is a prime,  Wang et al. \cite{wang2019pra} disproved the case $d=101$ of Problem \ref{prob5.1} or Conjecture \ref{conj5.1}
by finding a counterexample in $\mathbb{C}^{101}\otimes\mathbb{C}^{101}$.
The counterexample suggests that Problem \ref{prob5.1} (or Conjecture \ref{conj5.1}) may not be true.
But the necessary and sufficient conditions (Theorem \ref{tm5.1} or Corollary \ref{cor5.1}) show that
the case $d=5$ of Problem \ref{prob5.1} is true,
so this is a surprising result, and it suggests that the problem is more complex than expected.
Now we give a proof of Theorem \ref{tm5.1}.

\begin{proof}
Let $\mathcal{L}=\{(m_{i},n_{i})\}$ be a one-way LOCC distinguishable 4-GBS  or 5-GBS  set,
it is sufficient to show that the cardinality of the discriminant index set
$\{m^{-1}n|(m,n)\in\Delta\mathcal{L}, m^{-1}n\triangleq \infty \hbox{ if }m=0\}$ is less than 6.

It is known that all $C_{25}^{4}$(=12650) quadruples of GBSs can be classified into eight locally inequivalent
classes and the representative elements of these equivalence classes are
$H=\{ I, Z, Z^2, Z^3 \}, K=\{ I, Z, Z^2, X \},
L=\{ I, Z, Z^2, X^2 \}, \Gamma ^1_{1,1}, \Gamma ^1_{1,2}, \Gamma ^2_{2,1}, \Gamma ^1_{4,4}, \Gamma ^1_{2,2},$
where $ \Gamma ^a_{st} =  \{ I, Z, X^a, X^sZ^t \}$.
The first 6 sets are one-way LOCC distinguishable, while
the last 2 sets are one-way LOCC indistinguishable \cite{wang-yuan2021jmp}.
Since the set $\mathcal{L}$ is a one-way LOCC distinguishable,
it is LU-equivalent to one of the six 4-GBS sets  $H, K, L, \Gamma ^1_{1,1}, \Gamma ^1_{1,2}, \Gamma ^2_{2,1}$.
Because of the difference sets
$\Delta\mathcal{H}=\{(0,1),(0,2),(0,3),(0,4)\}$ and
$\Delta\mathcal{K}=\{(0,1),(0,2),(1,0),(1,4),(1,3),(0,4),(0,3),(4,0),\\(4,1),(4,2)\}$,
the cardinality of the corresponding discriminant index sets is 1 and 4 respectively.
Similarly, the cardinality of the discriminant index sets of $L, \Gamma ^1_{1,1}, \Gamma ^1_{1,2}$ and $\Gamma ^2_{2,1}$
is 4, 4, 5 and 4 respectively.
Therefore, the cardinality of the six sets is less than 6.

It is known that all $C_{25}^{5}$(=53130) 5-GBS sets can be divided into 21 locally inequivalent
classes and the representative elements of these equivalence classes are
\begin{widetext}
\begin{align*}
&H_1=\{ I, Z, Z^2, Z^3, Z^4 \},  H_2=\{ I, Z, Z^2, Z^3, X \};\ K_1=\{ I, Z, Z^2, X, X^2 \},  K_2=\{ I, Z, Z^2, X, X^3 \},\\
&K_3=\{ I, Z, Z^2, X, X^4 \},  K_4=\{ I, Z, Z^2, X, XZ \},K_5=\{ I, Z, Z^2, X, XZ^2 \},K_6=\{ I, Z, Z^2, X, X^2Z \},\\
&K_7=\{ I, Z, Z^2, X, X^2Z^2 \},  K_8=\{ I, Z, Z^2, X, X^3Z^2 \},K_9=\{ I, Z, Z^2, X, X^3Z^3 \},\\
&K_{10}=\{ I, Z, Z^2, X, X^3Z^4 \},K_{11}=\{ I, Z, Z^2, X, X^4Z \},  K_{12}=\{ I, Z, Z^2, X, X^4Z^2 \};\\
&L_1=\{ I, Z, Z^2, X^2, X^2Z \}, L_2=\{ I, Z, Z^2, X^2, X^2Z^2 \},L_3=\{ I, Z, Z^2, X^2, X^3Z \},\\
&L_4=\{ I, Z, Z^2, X^2, X^3Z^2 \};\Gamma _1= \{ I, Z, X, XZ^2, X^2Z \},  \Gamma _2=\{ I, Z, X, XZ^2, X^3Z^2 \},\\
&\Gamma _3=\{ I, Z, X, XZ^2, X^4Z \}.
\end{align*}
\end{widetext}
Among them nine sets
$H_1,H_2, K_1, K_2, K_4, K_{12}, L_1, L_4, \Gamma _2$
are one-way LOCC distinguishable (also two-way LOCC indistinguishable, see \cite{wang-yuan2021jmp}).
If the set $\mathcal{L}$ is a one-way LOCC distinguishable 5-GBS set,
then it is LU-equivalence to one of the nine 5-GBS sets.
Using a similar method as the 4-GBS sets,
we can deduce that the cardinality of the discriminant index sets of the nine 5-GBS sets is
1, 5, 5, 5, 5, 4, 5, 4 and 5 respectively.
So all of the nine 5-GBS sets have cardinality less than 6.
\end{proof}

In this section, we showed that, for each 4-GBS or 5-GBS set $\mathcal{L}=\{(m_{i},n_{i})\}$ residing in $\mathbb{C}^{5}\otimes\mathbb{C}^{5}$,
the one-way LOCC discrimination of $\mathcal{L}$ is completely determined by
the cardinality of $\it{discriminant\ index\ set}$ (=$\{m^{-1}n|(m,n)\in\Delta\mathcal{L}, m^{-1}n\triangleq \infty \hbox{ if }m=0\}$).
This simple and computable cardinality according to the following three facts:
(a) The cardinality is invariant under LU-equivalence.
(b) Each of the $C_{25}^{4}$(=12650) 4-GBS sets is LU-equivalent to one of the representative elements of eight equivalence classes,
and each of the $C_{25}^{5}$(=53130) 5-GBS sets is LU-equivalent to one of  the representative elements of 21 equivalence classes \cite{wang-yuan2021jmp}.
(c) The local discrimination of the representative elements can be determined by
\cite[Corollary 1]{band2011njp}, Theorem \ref{th3.1} and Yu et al.$^{,}$s result \cite[Lemma]{yu2015arxiv}.
Since \cite[Corollary 1]{band2011njp} and Theorem \ref{th3.1} are used to deal with the problem of one-way local discrimination,
and Yu et al.$^{,}$s result \cite[Lemma]{yu2015arxiv} can only be used to deal with the local discrimination of d-GBS sets in $\mathbb{C}^{d}\otimes\mathbb{C}^{d}$,
the problem of LOCC discrimination of 4-GBS sets in  $\mathbb{C}^{5}\otimes\mathbb{C}^{5}$ still remains unsolved.

Based on the above achievements in local discrimination of GBS sets in $\mathbb{C}^{4}\otimes\mathbb{C}^{4}$ and $\mathbb{C}^{5}\otimes\mathbb{C}^{5}$,
we may wish to consider the problem of local discrimination of GBS sets in $\mathbb{C}^{6}\otimes\mathbb{C}^{6}$,
specifically, according to the existing results, we only need to consider the cases of 4-GBS, 5-GBS and 6-GBS sets.
But it is difficult for us to solve this problem according to two basic facts:
(a) In $\mathbb{C}^{6}\otimes\mathbb{C}^{6}$, there are $C_{36}^{4}$(=58905) 4-GBS sets,
$C_{36}^{5}$(=376992) 5-GBS sets and $C_{36}^{6}$(=1947792) 6-GBS sets.
The number of sets to be considered is very large, so it is difficult to obtain comprehensive results.
(b) In order to deal with a large number of GBS sets, in general, it is necessary to classify these sets by using methods such as LU-equivalence.
However, such classification is difficult and has not been completed.
In short, this problem of local discrimination of GBS sets in $\mathbb{C}^{6}\otimes\mathbb{C}^{6}$ is interesting and valuable,
but it is difficult for us to solve it for the time being.

\section{Conclusions}
It is well known that,  for a  bipartite quantum system $\mathbb{C}^{d}\otimes\mathbb{C}^{d}$,
it is hard to locally distinguish a $k$-GBS set $(4\leq k\le d)$.
In this work, firstly by using Weyl commutation relation and difference set,
three practical sufficient conditions are given and several known results can be deduced as special cases of these conditions.
Based on LU equivalence,
it is shown that the three conditions are also necessary for local discrimination of GBS sets in $\mathbb{C}^{4}\otimes\mathbb{C}^{4}$,
and all 156 local indistinguishable standard 4-GBS sets are found (see Table II for details).
That is, the problem of local discrimination of GBS sets in $\mathbb{C}^{4}\otimes\mathbb{C}^{4}$ is completely solved.
In $\mathbb{C}^{5}\otimes\mathbb{C}^{5}$,
it is shown that a GBS set $\mathcal{L}$ is one-way LOCC indistinguishable if and only if
the cardinality of the discriminant index set $\{m^{-1}n|(m,n)\in\Delta\mathcal{L}, m^{-1}n\triangleq \infty \hbox{ if }m=0\}$ is 6.
Since the cardinality of the discriminant index set is easy to calculate,
and a general LOCC  has no advantage over one-way LOCC in distinguishing the 5-GBS sets,
the problem of LOCC discrimination of 5-GBS sets in $\mathbb{C}^{5}\otimes\mathbb{C}^{5}$ also are completely solved.
The necessary and sufficient condition also shows that
the case $d=5$ of the problem proposed by Wang et al. (Phys Rev A 99:022307, 2019) has a positive answer.
Despite these advances,
the problem of LOCC discrimination of 4-GBS sets in  $\mathbb{C}^{5}\otimes\mathbb{C}^{5}$ still remains unsolved.
It is hoped that the problem will be solved in the near future.

\begin{acknowledgments}
This work is supported by NSFC (Grant No. 11971151, 11901163) and Wuxi University Research Start-up Fund for Introduced Talents.
\end{acknowledgments}

\bibliographystyle{quantum}

\end{document}